\documentclass[paper=letter, fontsize=12pt]{article}
\usepackage[english]{babel} % English language/hyphenation
\usepackage{amsmath,amsfonts,amsthm} % Math packages
\usepackage[utf8]{inputenc}
\usepackage{float}
\usepackage{lipsum} % Package to generate dummy text throughout this template
\usepackage{blindtext}
\usepackage{graphicx}
\usepackage{caption}
\usepackage{subcaption}
\usepackage[sc]{mathpazo} % Use the Palatino font
\usepackage[T1]{fontenc} % Use 8-bit encoding that has 256 glyphs
\usepackage{microtype} % Slightly tweak font spacing for aesthetics
\usepackage[hmarginratio=1:1,top=32mm,columnsep=20pt]{geometry} % Document margins
\usepackage{multicol} % Used for the two-column layout of the document
\usepackage{booktabs} % Horizontal rules in tables
\usepackage{float} % Required for tables and figures in the multi-column environment - they need to be placed in specific locations with the [H] (e.g. \begin{table}[H])
\usepackage{hyperref} % For hyperlinks in the PDF
\usepackage{bookmark}
\usepackage{lettrine} % The lettrine is the first enlarged letter at the beginning of the text
\usepackage{paralist} % Used for the compactitem environment which makes bullet points with less space between them
\usepackage{abstract} % Allows abstract customization
\usepackage{titlesec} % Allows customization of titles

\renewcommand\thesection{\Roman{section}} % Roman numerals for the sections
\renewcommand\thesubsection{\Roman{subsection}} % Roman numerals for subsections

\titleformat{\section}[block]{\large\scshape\centering}{\thesection.}{1em}{} % Change the look of the section titles
\titleformat{\subsection}[block]{\large}{\thesubsection.}{1em}{} % Change the look of the section titles
\usepackage{fancyhdr} % Headers and footers
\newtheorem{corollary}{Corollary}
\title{\vspace{-15mm}\fontsize{24pt}{10pt}\selectfont\textbf{On expansions for the Black-Scholes prices and hedge parameters}} % Article title
\author{
{Jean-Philippe Aguilar }\\[2mm]
{\it BRED Banque Populaire, Modeling Department, 18 quai de la Râpée, Paris - 75012}\\[2mm]
{jean-philippe.aguilar@bred.fr}
}%\\ % Your name
\date{}

\providecommand{\keywords}[1]{\textbf{\textit{Key words---}} #1}
\newcommand{\res}{\mathrm{Res}}
\newcommand{\id}{\mathrm{d}}

\newtheorem{proposition}{Proposition}
\newtheorem{theorem}{Theorem}

\usepackage{mathtools}

\begin{document}
\maketitle % Insert title
%\tableofcontents
\pagestyle{headings}
\setcounter{page}{1}
\pagenumbering{arabic}
 \begin{abstract}
\noindent We derive new formulas for the price of the European call and put options in the Black-Scholes model, under the form of uniformly convergent series generalizing previously known approximations; these series are obtained by means of tools from multidimensional complex analysis. We also provide precise boundaries for the convergence speed and apply the results to the calculation of hedge parameters (Greeks).
\end{abstract}
% \keywords{Mellin transform, Multidimensional complex analysis, Fractional analysis, Stable distribution, Europen option pricing}


\thispagestyle{fancy} % All pages have headers and footers

\section{Introduction}

In capital markets, derivatives are financial instruments whose specificity is to be linked to another underlying instrument such as an asset, an index or an interest rate; they are traded by market participants as a way of transferring the risk on the latter instrument from a party to another \cite{IMF98}. Among most commonly traded derivatives are European call (resp. put) options, who give the right to the holder to buy (resp. sell) the underlying for an agreed amount, called strike price, at the expiration of the contract \cite{Willmott06,Hull08}. In order to evaluate the option itself, it is therefore necessary to model the dynamic of the underlying market; the most popular solution among practitioners is to use the Black-Scholes framework \cite{BS73}, where the log-returns of the underlying prices are assumed to be described by a geometric Brownian motion. Even if more sophisticated and realistic models have been introduced since then (for instance see \cite{KK16} and references therein), the Black-Scholes model remains the most commonly used by traders and financial engineers, because it admits an analytical solution, the Black-Scholes formula, which allows to express the European call price in terms of market parameters. 

Several ways exist to derive the Black-Scholes formula (eight different derivations are described in \cite{Andreasen98}, and even ten in the textbook \cite{Willmott09}), the most famous one being the resolution of the Black-Scholes partial differential equation (PDE). This PDE is classically obtained by a hedging argument combined with elementary stochastic calculus (which is the original derivation used by Black and Scholes in their seminal paper \cite{BS73}) but also as a limiting case of the binomial model \cite{Cox79} or by using the Capital Asset Pricing Model (CAPM) (see \cite{Andreasen98,Rouah} for instance). For non-PDE approaches, let us mention probabilistic derivations such as the martingale technique, which, although losing the direct connection with the hedging argument, provides a precise interpretation of the Black-Scholes formula in terms of risk-neutral probabilities \cite{Hull08,Nylsen}, or purely economic approaches like the "representative investor" as introduced by Rubinstein \cite{Rubinstein76}. 

In this paper, we would like to document a new derivation, based on a representation for the option price as a complex integral over a domain of $\mathbb{C}^2$ (obtained via suitable transforms of the Green function for the Black-Scholes PDE). This approach is fruitful because, evaluating this complex integral by means of multidimensional residues, we obtain series expansions for the price of the call and put options, which turn out to be simple and fast convergent; these expansions recover existing approximations, that were known in some very specific market configurations (when the asset is at-the-money forward, it is easy to expand the Black-Scholes formula as a power series of the market volatility \cite{BS94}). Moreover, they are unconditionally and uniformly convergent, which would not be the case with a naive Taylor expansion: as already noticed by Estrella in \cite{Estrella95}, the Black-Scholes formula mixes two components of strongly different natures  (the logarithmic function, possessing an expansion converging very fast but only for a small range of arguments, and the normal distribution, whose expansion  converges on the whole real axis but with fewer level of accuracy) resulting in situations where, for a plausible range of parameter values, the Taylor series for the Black-Scholes price diverges. 

The paper is organized as follow. In section 2, we first write down the option price as a double complex integral, and evaluate it by help of residue theory in $\mathbb{C}^2$. This will allow us to express the call option price under the form of a simple series \eqref{BS_series}, which will be refined into another series \eqref{BS_series_3} exhibiting more clearly its financial meaning; last for this section, we prove the exponential convergence to the option price. In section 3, we apply \eqref{BS_series_3} to the specific at-the-money configuration, and we compute the series expansions for the hedge parameters (known as "Greek letters"), which measure the option's sensitivity to the variations of market parameters. After the conclusive section and for reader's convenience, we have equipped the paper with an appendix summing up (without proof) the main concept of multidimensional complex analysis used in the present article.

\section{Pricing formulas}

Let us start by introducing some usual financial notations. In all of the following, $T$ will be some positive real number and $S$ will represent the market price of a financial instrument (more precisely, $\{ S_t \}_{t\in[0,T]}$ will be a stochastic process on a filtered probability space $(\Omega,\mathcal{F},\mathbb{P})$). $C$ will denote the price of the European call option on $S$, with maturity $T$ and strike $K\in\mathbb{R}$. As for multidimensional complex analysis notations, we will denote the vectors in $\mathbb{C}^2$ by $\underline{u}=\begin{bmatrix} u_1 \\ u_2 \end{bmatrix}$, $u_1,u_2\in\mathbb{C}$; we will also use the standard wedge symbol for differential forms, which has the properties \cite{Griffiths78}:
\begin{equation}
\id u_1 \, \wedge \, \id u_2 \, = \, - \, \id u_2 \, \wedge \, \id u_1 \,\, \hspace*{2cm} \,\, \id u_1 \, \wedge \, \id u_1 \, = \, 0
\end{equation}
and we will denote $\id \underline{u} := \id u_1 \wedge \id u_2$.

\subsection{The call price as a complex integral}

In the Black-Scholes model, it is assumed that, under a certain measure, the instantaneous variations of $S$ are driven by a geometric Brownian motion of drift $r\in\mathbb{R}$ (risk-free interest rate) and volatility $\sigma >0$ (representing the volatility of returns of the underlying asset). Using It\^o calculus, one can show that the call price $C$ obeys the Black-Scholes PDE with terminal condition: 
\begin{align}\label{BS_Equation}
\left\{
\begin{aligned}
 & \frac{\partial C}{\partial t} \, + \, \frac{1}{2} \sigma^2 S^2\frac{\partial^2 C}{\partial S^2} \, + \, r S \frac{\partial C}{\partial S} \, - \, rC \, = \, 0  \hspace*{1cm}  t\in[0,T]  \\
 & C \, = \, \max\{ S-K , 0 \} \, := \,  [S-K]^+   \hspace*{1.7cm}     t = T
\end{aligned}
\right.
\end{align}
The solution to \eqref{BS_Equation} can be written under the form of a convolution of the heat kernel with the modified terminal payoff:
\begin{equation}\label{BS_GreenForm}
C \, = \, e^{-r\tau} \, \int\limits_{-\infty}^{+\infty} [Se^{(r-\frac{\sigma^2}{2})\tau+y}-K]^{+} \, \frac{1}{\sigma\sqrt{2\pi\tau}}e^{-\frac{y^2}{2\sigma^2 \tau}} \, \id y
\end{equation}
where $\tau = T - t$. The presence of the heat kernel is a consequence of the fact that, under a suitable change of variables, the Black-Scholes PDE resumes to the heat equation (see full details in \cite{Willmott06} for instance): 
\begin{equation}\label{heat_equation}
\frac{\partial W}{\partial \tau} \, - \, \frac{\sigma^2}{2}\frac{\partial^2 W}{\partial x^2} \, = \, 0
\end{equation}
where $C=e^{-r\tau}W$.

Let us now introduce the notations for the volatility $z$, the forward strike price $F$ and the log-forward moneyness $k$:
\begin{equation}\label{notations}
% \tau \, := \, T - t
% \hspace*{1cm}
z \, := \, \sigma\sqrt{\tau}
\hspace*{1cm}
F \, := \, Ke^{-r\tau}
\hspace*{1cm}
k \, := \, \log\frac{S}{F} \, = \, \log\frac{S}{K} + r\tau
\end{equation}
Making basic manipulations on the integral \eqref{BS_GreenForm} yields the Black-Scholes formula which, in our system of notations, reads:
\begin{equation}\label{BlackScholesFormula}
C \, = \, S N \left( \frac{k}{z} + \frac{z}{2}  \right) \, -  F N \left( \frac{k}{z} - \frac{z}{2}  \right) \, = \, F \left[ e^{k} N \left( \frac{k}{z} + \frac{z}{2}  \right) - N \left( \frac{k}{z} - \frac{z}{2}  \right) \right]
\end{equation}
where $N(.)$ is the standard normal cumulative distribution function. Our purpose is to provide an alternative to the Black-Scholes formula \eqref{BlackScholesFormula}, under the form of a simple and rapidly convergent series. To that extent, let us start by expressing the call option price \eqref{BS_GreenForm} as an integral over a domain of $\mathbb{C}^2$.  
\begin{proposition}
Let $P\subset\mathbb{C}^2$ be the polyhedra $P:=\{ \underline{t}\in\mathbb{C}^2 \, , \, Re(2t_1+t_2) > 2 \, , \, 0 < Re(t_2) < 1 \}$. Then, for any $\underline{c}\in P$,
\begin{equation}\label{BS_4}
C \, = \, F  \int\limits_{\underline{c}+i\mathbb{R}^2}  \, (-1)^{-t_2} \, 2^{\frac{1}{2}-t_1} \frac{\Gamma(t_2)\Gamma(1-t_2)\Gamma(-2+2t_1+t_2)}{\Gamma(t_1+\frac{1}{2})} \left( \frac{z^2}{2} - k  \right)^{2-2t_1-t_2} z^{2t_1-1}  \frac{\id \underline{t}}{(2i\pi)^2}
%\\  z^{2t_1-1}  \frac{\id \underline{t}}{(2i\pi)^2}
\end{equation}
\end{proposition}
\begin{proof}
With our notations \eqref{notations}, we can re-write the call price \eqref{BS_GreenForm} as:
\begin{align}\label{BS_2}
C & = \frac{F}{\sqrt{2\pi}} \,  \int\limits_{\frac{z^2}{2}-k}^{\infty} \, (e^{k - \frac{z^2}{2} +y} -1 ) \, \frac{1}{z} e^{-\frac{y^2}{2z^2}} \id y \, 
\end{align}
Let us introduce a Mellin-Barnes representation for the heat kernel-term in \eqref{BS_2} (see \eqref{Cahen} in Appendix, and \cite{Flajolet95,Erdélyi54} or any monograph on integral transforms):
\begin{equation}
\frac{1}{z} e^{-\frac{y^2}{2z^2}} \, = \, \frac{1}{z} \, \int\limits_{c_{1}-i\infty}^{c_{1}+i\infty} \Gamma(t_1)  \left( \frac{y^2}{2z^2}  \right)^{-t_1} \, \frac{\id t_1}{2i\pi} \hspace*{1cm} ( c_{1}> 0 )
\end{equation}
We thus have:
\begin{equation}
C \, = \, \frac{F}{\sqrt{2\pi}} \, \int\limits_{c_{1}-i\infty}^{c_{1}+i\infty} 2^{t_1} \Gamma(t_1) \, \int\limits_{\frac{z^2}{2}-k}^{\infty} \, (e^{k - \frac{z^2}{2} +y} -1 ) \,y^{-2 t_1} \, \id y \, z^{2t_1-1} \, \frac{\id t_1}{2i\pi} 
\end{equation}
Integrating by parts in the $y$-integral yields:
\begin{equation}
C \, = \, \frac{F}{\sqrt{2\pi}} \, \int\limits_{c_{1}-i\infty}^{c_{1}+i\infty} 2^{t_1} \frac{\Gamma(t_1)}{2t_1-1} \, \int\limits_{\frac{z^2}{2}-k}^{\infty} \, e^{k - \frac{z^2}{2} +y} \,y^{1-2 t_1} \, \id y \, z^{2t_1-1} \, \frac{\id t_1}{2i\pi} 
\end{equation}

\noindent
Let us introduce another Mellin-Barnes representation for the remaining exponential term (again, see \eqref{Cahen} in Appendix):
\begin{equation}
e^{k-\frac{z^2}{2}+y} \, = \, \int\limits_{c_{2}-i\infty}^{c_{2}+i\infty} \, (-1)^{-t_2} \, \Gamma(t_2)  \left( k-\frac{z^2}{2}+y   \right)^{-t_2} \, \frac{\id t_2}{2i\pi} \hspace*{1cm} ( c_{2}> 0 )
\end{equation}
Therefore the call price is:
\begin{multline}\label{BS_3}
C \, = \, \frac{F}{\sqrt{2\pi}} \times 
\int\limits_{c_{1}-i\infty}^{c_{1}+i\infty} \int\limits_{c_{2}-i\infty}^{c_{2}+i\infty}  \, (-1)^{-t_2} \, \frac{2^{t_1}}{2t_1-1} \, \Gamma(t_1) \Gamma(t_2) \, \int\limits_{\frac{z^2}{2}-k}^{\infty} y^{1-2t_1} \\ \left( k-\frac{z^2}{2}+y   \right)^{-t_2} \id y  \, z^{2t_1-1} \, \frac{\id t_1}{2i\pi} \wedge \frac{\id t_2}{2i\pi}
\end{multline}
The $y$-integral is a particular case of B\^eta-integral \cite{Abramowitz72} and equals:
\begin{equation}
\int\limits_{\frac{z^2}{2}-k}^{\infty} y^{1-2t_1} \left( k-\frac{z^2}{2}+y   \right)^{-t_2} \id y \, = \, \left( \frac{z^2}{2} - k  \right)^{2-2t_1-t_2} \frac{\Gamma(1-t_2)\Gamma(-2+2t_1+t_2)}{\Gamma(2t_1-1)}
\end{equation}
and converges on the conditions $Re(t_2)<1$ and $Re(2t_1+t_2)>2$; plugging into \eqref{BS_3}, using the Gamma function functional relation $(2t_1-1)\Gamma(2t_1-1)=\Gamma(2t_1)$ and the Legendre duplication formula for the Gamma function \cite{Abramowitz72}:
\begin{equation}
\frac{\Gamma(t_1)}{\Gamma(2t_1)} \, = \, 2^{1-2t_1} \sqrt{\pi} \, \frac{1}{\Gamma(t_1+\frac{1}{2})}
\end{equation}
we obtain the integral \eqref{BS_4}.
\end{proof}

\subsection{Residue summation}

Now we compute the double integral \eqref{BS_4} by means of summation of multidimensional residues.
\begin{theorem}
Let $Z:=\frac{z}{\sqrt{2}}$ (normalized volatility); then the Black-Scholes price of the European call is:
\begin{equation}\label{BS_series}
C \, = \, \frac{F}{2} \sum\limits_{\substack{n = 0 \\ m = 1}}^{\infty} \, \frac{(-1)^n}{n!\Gamma(1+\frac{m-n}{2})}  \,  \left( Z^2 - k  \right)^{n} Z^{m-n}
\end{equation}
\end{theorem}
\begin{proof}
Let $\omega$ be the \textit{complex differential 2-form}
\begin{multline}\label{omega}
\omega \, := \,  (-1)^{-t_2} \, 2^{\frac{1}{2}-t_1} \, \frac{\Gamma(t_2)\Gamma(1-t_2)\Gamma(-2+2t_1+t_2)}{\Gamma(t_1+\frac{1}{2})} \left( \frac{z^2}{2} - k  \right)^{2-2t_1-t_2} z^{2t_1-1} \, \frac{\id t_1}{2i\pi} \wedge \frac{\id t_2}{2i\pi}
% \\ \, \frac{\id t_1}{2i\pi} \wedge \frac{\id t_2}{2i\pi}
\end{multline}
so that we can write the call price \eqref{BS_4} under the form:
\begin{equation}\label{V_form}
C \, = \, F \, \int\limits_{\underline{c} + i\mathbb{R}^2} \, \omega \hspace*{1cm} (\underline{c} \in P)
\end{equation}

\noindent
This complex integral can be performed by means of summation of $\mathbb{C}^2$-residues, by virtue of a multidimensional analogue to the residue theorem valid for this specific class of integrals (see \eqref{Tsikh_residue} and references in Appendix). Indeed, the characteristic quantity (see definition \eqref{Delta_n} in appendix) associated to the differential form \eqref{omega} is:
 
\begin{figure}[t]
\centering
\includegraphics[scale=0.5]{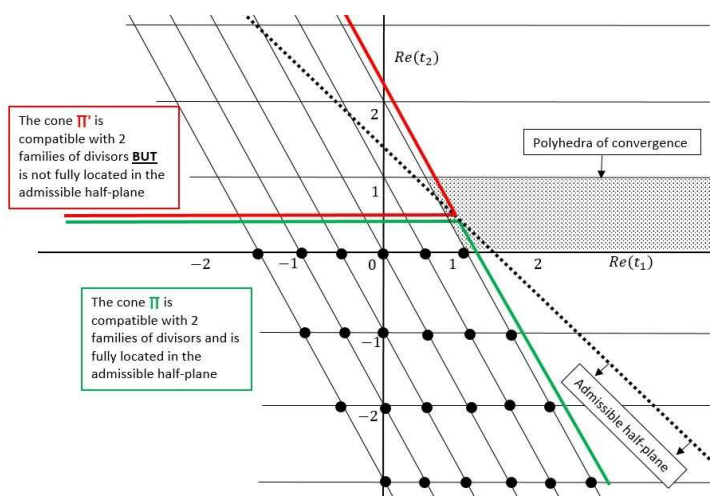}
\caption{The divisors $D_1$ (oblique lines) are induced by the $\Gamma(-2+2t_1+t_2)$ term, and $D_2$ (horizontal lines) by the $\Gamma(t_2)$ term. The intersection set $D_1 \cap D_2$ (the dots), located in the compatible green cone $\Pi$, gives birth to residues whose sum in the whole cone equals the integral \eqref{V_form}.}
\label{fig4}
\end{figure}

\noindent 
\begin{equation}
\Delta \, = \, 
\begin{bmatrix}
2 - 1 \\ 1 - 1 + 1
\end{bmatrix}
\, = \, 
\begin{bmatrix}
1 \\ 1
\end{bmatrix}
\end{equation}
and the admissible half-plane is 
\begin{equation}
\Pi_{\Delta} \, := \, 
\left\{
\underline{t}\in\mathbb{C}^2, \, Re (  \Delta \, . \, \underline{t}  ) \, < \,  \Delta . \underline{c} 
\right\}
\end{equation}
in the sense of the Euclidean scalar product. This half-plane is therefore located under the line
\begin{equation}
t_2 \, = \, -t_1 \, + \, c_{1} \, + \, c_{2}
\end{equation}
In this half-plane, the cone $\Pi$ as shown of fig. \ref{fig4} and defined by 
\begin{equation}
 \Pi \, : = \, \left\{ \underline{t}\in\mathbb{C}^2 \, , \, Re(t_2) \leq 0 \, , \, Re(2t_1+t_2) \leq 2       \right\}
\end{equation}
contains and is compatible with the two families of divisors
\begin{equation}
\left\{
\begin{aligned}
& D_1 \, = \, \left\{\underline{t}\in\mathbb{C}^2, -2+2t_1 + t_2 = -n_1 \,\, , \,\,\, n_1 \in\mathbb{N} \right\}
\\
& D_2 \, = \, \left\{\underline{t}\in\mathbb{C}^2, t_2 = -n_2 \,\, , \,\,\, n_2 \in\mathbb{N} \right\}
\end{aligned}
\right.
\end{equation}
induced by $\Gamma(-2+2t_1+t_2)$ and $\Gamma(t_2)$ respectively. This configuration is shown on fig. \ref{fig4}. To compute the residues associated to every element of the singular set $ D_1 \cap D_2$, we change the variables:
\begin{equation}
\left\{
\begin{aligned}
& u_1 \, := \, -2 + 2t_1 + t_2 \\
& u_2 \, := \, t_2
\end{aligned}
\right.
\longrightarrow
\left\{
\begin{aligned}
& t_1 \, = \, \frac{1}{2} (2 + u_1 - u_2) \\
& t_2 \, = \, u_2 \\ 
& dt_1 \wedge dt_2 \, = \, \frac{1}{2} \, du_1 \wedge du_2
\end{aligned}
\right.
\end{equation}
so that in this new configuration $\omega$ reads
\begin{equation}
\omega \, = \, \frac{1}{2} (-1)^{-u_2} \, 2^{\frac{u_2-u_1-1}{2}} \, \frac{\Gamma(u_2)\Gamma(1-u_2)\Gamma(u_1)}{\Gamma(1+\frac{u_1-u_2+1}{2})} \left( \frac{z^2}{2} - k  \right)^{-u_1} z^{u_1-u_2+1} \, \frac{\id u_1}{2i\pi } \wedge \frac{\id u_2}{2i\pi}
\end{equation}
With this new variables, the divisors $D_1$ and $D_2$ are induced by the $\Gamma(u_1)$ and $\Gamma(u_2)$ functions in $(u_1,u_2)=(-n,-m)$, $n,m\in\mathbb{N}$. From the singular behavior of the Gamma function around a singularity \eqref{sing_Gamma}, we can write
\begin{multline}
\omega \, \underset{(u_1,u_2)\rightarrow (-n,-m)}{\sim} \, 
\frac{1}{2} (-1)^{-u_2} \, \frac{(-1)^{n+m}}{n!m!} \, 2^{\frac{u_2-u_1-1}{2}} \, \frac{\Gamma(1-u_2)}{\Gamma(1+\frac{u_1-u_2+1}{2})} \left( \frac{z^2}{2} - k  \right)^{-u_1} \\ z^{u_1-u_2+1} \, \frac{\id u_1}{2i\pi (u_1+n)} \wedge \frac{\id u_2}{2i\pi (u_2+m)}
\end{multline}
and therefore the residues are, by the Cauchy formula:
\begin{equation}
\res_{(-n,-m)} \, \omega \, = \,  (-1)^{n} \, 2^{\frac{n-m-1}{2}-1} \, \frac{1}{n!\Gamma(1+\frac{m-n+1}{2})}  \,  \left( \frac{z^2}{2} - k  \right)^{n} z^{m-n+1}
\end{equation}
By virtue of the residue theorem \eqref{Tsikh_residue}, the sum of the residues in the whole cone equals the integral \eqref{V_form}:
\begin{equation}\label{V1-series}
V \, = \, F \sum\limits_{n,m=0}^{\infty} \, (-1)^{n} \, 2^{\frac{n-m-1}{2}-1} \, \frac{1}{n!\Gamma(1+\frac{m-n+1}{2})}  \,  \left( \frac{z^2}{2} - k  \right)^{n} z^{m-n+1}
\end{equation}
\noindent
We can further simplify by changing the index $m\rightarrow m+1$ and introducing $Z:=\frac{z}{\sqrt{2}}= \frac{\sigma\sqrt{\tau}}{\sqrt{2}}$, and we finally obtain the series \eqref{BS_series}.
\end{proof}
Formula \eqref{BS_series} is a new representation for the Black-Scholes price of the European call option, and can be very easily implemented in order to make precise calculations; one may note that this expansion is surprisingly simple and compact, when compared to the Black-Scholes formula \eqref{BlackScholesFormula}. 
In table \ref{fig:BS_Series}, we test the precision of \eqref{BS_series}, by comparing the prices obtained via the usual Black-Scholes formula \eqref{BlackScholesFormula} to the prices given by various truncations ($n=0\dots N_{max}$ and $m=1\dots M_{max}$) of our series \eqref{BS_series}. One observes that the convergence is extremely fast: typically, $N_{max}=M_{max}=10$ are enough to obtain a precision of $10^{-7}$, and, when the option is deeply in or out-of-the-money, the convergence is slightly slower but remains very efficient (notably in the ITM region). Let us also remark that, when the asset is at-the-money forward (that is, when $S=F$ and therefore when $k=0$), then \eqref{BS_series} resumes to the power series:
\begin{equation}\label{Brenner_1}
C \, = \, \frac{S}{2} \sum\limits_{\substack{n = 0 \\ m = 1}}^{\infty} \, \frac{(-1)^n}{n!\Gamma(1+\frac{m-n}{2})}  \,  Z^{m+n}
\end{equation}

\begin{table}[h!]
\centering
\begin{scriptsize}
\begin{tabular}{|l|c|c|c|c|}
\hline
Market configuration & BS formula & $N_{max}=M_{max}=5 $ & $N_{max}=M_{max}=10 $ & $N_{max}=M_{max}=20 $  \\
\hline 
Deep OTM ($S=3000$) & 25.8385546  &  14.6150001  & 25.9147783 & 25.8385533 \\
OTM ($S=3800$)  & 235.5135954 & 235.5112726  & 235.5135954 & 235.5135954  \\
ATM forward  & 315.4523494 & 315.4501517 &  315.4523494 & 315.4523494 \\
ITM ($S=4200$) & 458.7930654 & 458.7883563 & 458.7930654 & 458.7930654 \\
Deep ITM ($S=5000$) & 1093.1653246 & 1091.3521829 &  1093.1662581 & 1093.1653246 \\
  \hline
\end{tabular}
\end{scriptsize}
\caption{Convergence of the series \eqref{BS_series} to the usual Black-Scholes formula, for various market configurations. In all cases, $K=4000$, $r=1\%$, $\sigma= 20 \%$ and $\tau$=1Y.}
\label{fig:BS_Series}
\end{table}

However, despite its simplicity, formula \eqref{BS_series} does not exhibit clearly the ordering in powers of $Z$ and the simplifications that can be operated; for instance in the power series \eqref{Brenner_1} all even powers of $Z$ vanish (as expected from the Taylor series for the normal distribution). Let us therefore refine \eqref{BS_series} into an expansion with an even more tractable structure.
% We may remark than, by forcing the factorization by $Ke^{-r\tau}$ in the Black-Scholes formula \eqref{BlackScholesFormula} we obtain:
% \begin{equation}\label{BlackScholes_factorized}
% V(F,k,Z) \, = \, F \left[ e^{k}N \left(\frac{k}{z}+\frac{1}{2}z \right) \, - \, N \left(\frac{k}{z}-\frac{1}{2}z \right) \right]
% \end{equation}
% and therefore the series \eqref{BS_series} is a surprisingly simple expansion for the r.h.s. of \eqref{BlackScholes_factorized}. Let us now derive other forms for the series \eqref{BS_series}.
\begin{corollary}
Let the $\varphi_J$ be the real functions defined by
\begin{equation}
\varphi_J(x) \, = \, \sum_{n=0}^{J-1} \, \frac{(-1)^n}{n!\Gamma(1+\frac{J}{2}-n)}(1-x)^n \hspace*{1cm} J= 1,2,\dots
\end{equation}
Then the Black-Scholes price of the European call is
\begin{equation}\label{BS_Series_2}
C \, = \, \frac{F}{2} \sum_{J=1}^\infty \, Z^J \varphi_J \left( \frac{k}{Z^2} \right)
\end{equation}
\end{corollary}
\begin{proof}
Let $J:=m+n$ in the series \eqref{BS_series}, $J\geq 1$. As $m\geq 1$, this implies $J-n \geq 1$, that is, the upper bound of the $n$ summation within this new configuration is $n\leq J-1$. This means that we are performing the summation along oblique lines instead of horizontal or vertical ones (see Fig \ref{fig5}); we can therefore introduce the $J$-th partial sum (recall that now $m-n=J-2n$)
\begin{align}
C_J  \, & :=  \frac{F}{2} \sum_{n=0}^{J-1} \frac{(-1)^n}{n!\Gamma(1+\frac{J}{2}-n)} (Z^2-k)^n Z^{J-2n} \\
               &  =  \frac{F}{2} \, Z^J \, \sum_{n=0}^{J-1} \frac{(-1)^n}{n!\Gamma(1+\frac{J}{2}-n)} \left(1-\frac{k}{Z^2} \right)^n           
\end{align}
and summing over all the oblique lines $J\geq 1$ yields formula \eqref{BS_Series_2}
\begin{figure}[h]
\centering
\includegraphics[scale=0.4]{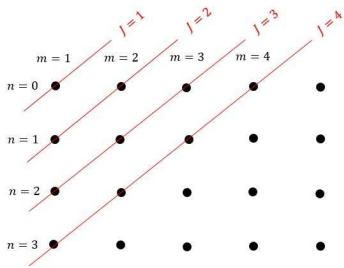}
\caption{The summation \eqref{BS_series} can be performed on oblique lines ($J=1,2,\dots$) instead of vertical lines $m=1,2,\dots$ or horizontal lines $n=0,1,\dots$.}
\label{fig5}
\end{figure}
\end{proof}

\begin{proposition}\label{proposition_2j}
For any integer $j\geq 1$,
\begin{equation}
Z^{2j} \varphi_{2j} \left( \frac{k}{Z^2} \right) \, = \, \frac{k^j}{j!}
\end{equation}
\end{proposition}
\begin{proof}
By definition of $\varphi_{2j}$, we have
\begin{equation}
\varphi_{2j} \left( \frac{k}{Z^2} \right) \, = \,  \sum_{n=0}^{2j-1} \frac{(-1)^n}{n!\Gamma(j+1-n)} \left(1-\frac{k}{Z^2} \right)^n
\end{equation}
As the Gamma function in the denominator is infinite for negative integers, all terms after $n=j$ vanish and therefore the sum can be written as:
\begin{align}
\varphi_{2j} \left( \frac{k}{Z^2} \right) \, & = \,  \sum_{n=0}^{j} \frac{(-1)^n}{n!\Gamma(j+1-n)} \left(1-\frac{k}{Z^2} \right)^n \\
 & = \, \sum_{n=0}^{j} \frac{1}{n!(j-n)!} \left(\frac{k}{Z^2} -1 \right)^n \\ 
 & = \, \frac{1}{j!} \frac{k^j}{Z^{2j}}
\end{align}
in virtue of the binomial theorem.

% We start the recurrence at $j=1$:
% \begin{equation}
% Z^2 \varphi_{2j} \left( \frac{k}{Z^2} \right) = Z^2 \left[\frac{1}{\Gamma(2)} - \frac{1}{\Gamma(1)}\left( 1-\frac{k}{Z^2} \right) \right] =k
% \end{equation}
\end{proof}

\begin{theorem}
The Black-Scholes price of the European call option is:
\begin{equation}\label{BS_series_3}
C \, = \, \frac{1}{2}(S-F) \, + \, \frac{F}{2} \, \sum_{\substack{j\geq 0 \\ n\leq 2j}} \, Z^{2j+1} \, \frac{(-1)^n}{n!\Gamma(\frac{3}{2}+j-n)}  \, \left( 1-\frac{k}{Z^2} \right)^n
\end{equation}
\end{theorem}
\begin{proof}
Let us write
\begin{align}
C  & \, = \, \frac{F}{2}  \left[ \sum_{j=1}^\infty Z^{2j} \varphi_{2j} \left( \frac{k}{Z^2}  \right) \, + \, \sum_{j=0}^\infty Z^{2j+1} \varphi_{2j+1} \left( \frac{k}{Z^2}  \right) \right] \\
 & \, = \, \frac{F}{2} \left[ (e^k-1) \, + \, \sum_{j=0}^\infty Z^{2j+1} \varphi_{2j+1} \left( \frac{k}{Z^2}  \right) \right]
\end{align}
where we have used proposition \ref{proposition_2j} to perform the first sum. Recalling that $e^{k} = \frac{S}{F} $ and simplifying yields eq. \eqref{BS_series_3}
\end{proof}

Formula \eqref{BS_series_3} is less compact than formula \eqref{BS_series}, but may be more suitable for practical applications. Notably, it immediately follows from the call-put parity $C-P = S - F$ that the Black-Scholes put is given by:
\begin{equation}
P \, = \, \frac{1}{2}(F-S) \, + \, \frac{F}{2} \, \sum_{\substack{j\geq 0 \\ n\leq 2j}} \, Z^{2j+1} \, \frac{(-1)^n}{n!\Gamma(\frac{3}{2}+j-n)}  \, \left( 1-\frac{k}{Z^2} \right)^n
\end{equation}
In \eqref{BS_firstterms} we write down the expansion \eqref{BS_series_3} up to order $Z^5$:
%note that the value of gamma functions at half-integers are actually all known analytically, because $\Gamma(1/2)=\sqrt{\pi}$ and $\Gamma(z+1)=z\Gamma(z)$  \cite{Abramowitz72}:
{\footnotesize
\begin{align}\label{BS_firstterms}
& C  \, =  \, \hspace*{0.2cm} \frac{1}{2} (S-F) \\
&   \, + \, \frac{F}{2} \, \left[ \hspace*{0.3cm} Z \,\,  \frac{1}{\Gamma(\frac{3}{2})}   \right. \nonumber  \\
&      \, \hspace*{0.8cm} + \, Z^3 \left( \frac{1}{\Gamma(\frac{5}{2})} - \frac{1}{\Gamma(\frac{3}{2})} \left( 1 - \frac{k}{Z^2} \right) + \frac{1}{2\Gamma(\frac{1}{2})} \left( 1 - \frac{k}{Z^2} \right)^2 \right)  \nonumber \\
&   \, \hspace*{0.8cm} + \, Z^5 \left( \frac{1}{\Gamma(\frac{7}{2})} - \frac{1}{\Gamma(\frac{5}{2})} \left( 1 - \frac{k}{Z^2} \right) + \frac{1}{2\Gamma(\frac{3}{2})} \left( 1 - \frac{k}{Z^2} \right)^2 - \frac{1}{6\Gamma(\frac{1}{2})} \left( 1 - \frac{k}{Z^2} \right)^3 \right. \nonumber \\
& \, \hspace*{8cm} \left. + \frac{1}{24\Gamma(-\frac{1}{2})} \left( 1 - \frac{k}{Z^2} \right)^4 \right) + \cdots \nonumber 
\end{align}
}
Note that the values of the Gamma function involved in \eqref{BS_series_3} are actually known analytically, as a consequence of the functional relation
\begin{equation}
 \Gamma\left(\frac{3}{2}+j-n\right) \, = \,  \left(\frac{1}{2}+j-n\right) \Gamma\left(\frac{1}{2}+j-n\right)   
\end{equation}
and of the particular values of the Gamma function at half-integers \cite{Abramowitz72}:
\begin{equation}
\Gamma\left(\frac{1}{2}+j-n\right) \, = \,
\left\{
\begin{aligned}
& \frac{(2(j-n))!}{4^{j-n}(j-n)!} \sqrt{\pi}  & \,\, (j-n \geq 0) \\
& \frac{(-4)^{|j-n|}|j-n|!}{(2|j-n|)!} \sqrt{\pi} & \,\,  (j-n < 0)
\end{aligned}
\right.
\end{equation}
Therefore, the evaluation of these terms resumes to computations of factorials, which are easily carried out without needing sophisticated tools. Moreover, only few terms of the series \eqref{BS_series_3} are needed to obtain a very precise approximation of the Black-Scholes formula \eqref{BlackScholesFormula}, as demonstrated by figure \ref{fig6} and by the boundary established in the next subsection.

\begin{figure}[h]
\centering
\includegraphics[scale=0.4]{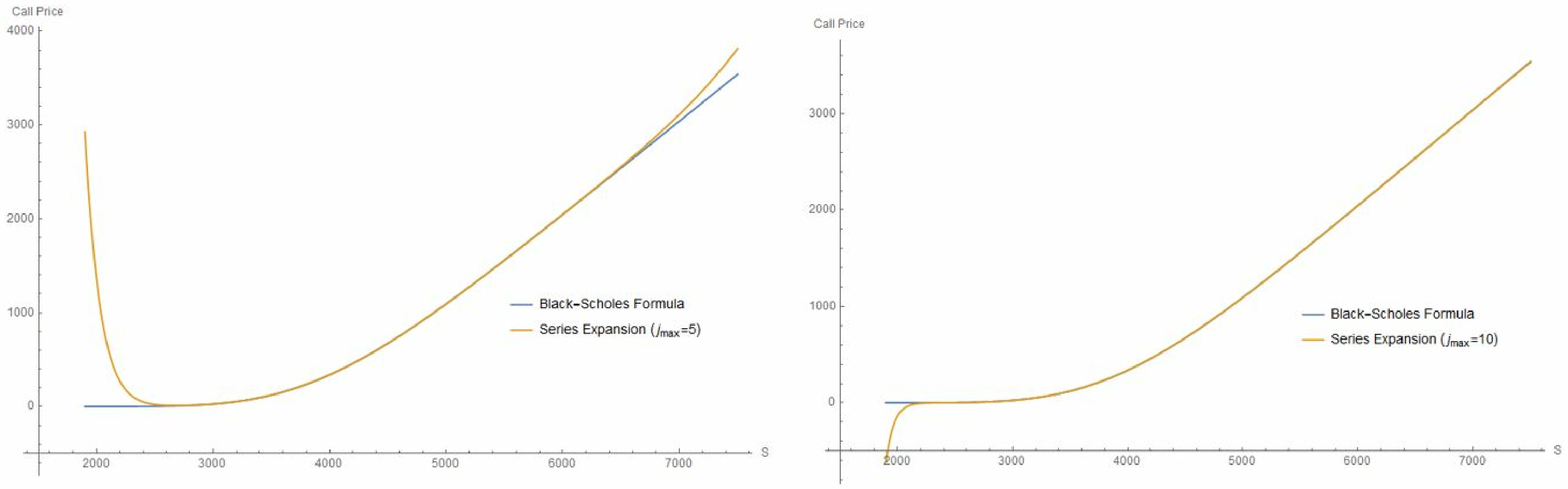}
\caption{Comparisons between the prices obtained via the Black-Scholes formula \eqref{BlackScholesFormula} and the series formula \eqref{BS_series_3} truncated at $j_{max}=5$ (left graph) or $j_{max}=10$ (right graph), for $K=4000$, $r=1\%$, $\sigma=20\%$ and $\tau=$1Y. One can observe that, for a wide range of prices (approximately $S\in[2500,6500]$), it is sufficient to truncate the series \eqref{BS_series_3} to $j_{max}=5$ to reproduce the Black-Scholes formula. Increasing the number of terms to $j_{max}=10$ widens the range of convergence to even more deeply in or out-of-the money situations.}
\label{fig6}
\end{figure}

\subsection{Speed of convergence}

% We can assume for simplicity that $Z<1$ (that is, $\sigma\sqrt{\tau}<\sqrt{2}$, which turns out to be the case for any maturity $\tau<5Y$ and for a typical market volatility $\sigma = 20\%$). 
\begin{proposition}\label{convergence}
Let $\alpha:= \max\left( 1, \left\vert 1 - \frac{k}{Z^2} \right \vert  \right)$. Then the generic term of the series \eqref{BS_series_3} is bounded (uniformly in $n$) by
\begin{equation}\label{Speed}
\left\vert Z^{2j+1} \, \frac{(-1)^n}{n!\Gamma(\frac{3}{2}+j-n)}  \, \left( 1-\frac{k}{Z^2} \right)^n \right\vert \,  \leq \, \frac{Z}{\sqrt{\pi}} \frac{(\alpha Z)^{2j}}{( \lfloor \frac{j}{2} \rfloor +1) !}
\end{equation}
% and therefore the series converges normally to the option price.
\end{proposition}
\begin{proof}
Let $R_{j,n}:=\frac{(-1)^n}{n!\Gamma(\frac{3}{2}+j-n)}$. We have:
\begin{equation}\label{Speed1}
\left\vert Z^{2j+1} \, \frac{(-1)^n}{n!\Gamma(\frac{3}{2}+j-n)}  \, \left( 1-\frac{k}{Z^2} \right)^n \right\vert \,  \leq \, Z^{2j+1} \left\vert R_{j,n} \right\vert \, \alpha^{2j}
\end{equation}
Derivating with respect to $n$, it is easy to see that $R_{j,n}$ is maximal on the line $n= \lfloor \frac{j}{2} \rfloor +1$:
\begin{equation}\label{Speed2}
\left\vert R_{j,n}  \right\vert \, \leq \, \left\vert R_{j,\lfloor \frac{j}{2} \rfloor +1} \right\vert \, = \, \frac{1}{(\lfloor \frac{j}{2} \rfloor +1)!\Gamma(\frac{1}{2}+j- \lfloor \frac{j}{2} \rfloor)}
\end{equation}
As $x\rightarrow\Gamma(\frac{1}{2}+x)$ grows for $x>0$, we have $\Gamma\left(\frac{1}{2}+j- \lfloor \frac{j}{2} \rfloor \right) \, \geq \, \Gamma\left(  \frac{1}{2} \right) \, = \, \sqrt{\pi}$ which terminates the proof.
\end{proof}

Therefore, if we wish to attain a precision of $\epsilon$ in the series representation \eqref{BS_series_3}, we just need to find the integer $j_\epsilon$ such as
\begin{equation}\label{Speed3}
\frac{Z}{\sqrt{\pi}} \frac{(\alpha Z)^{2j_\epsilon}}{( \lfloor \frac{j_\epsilon}{2} \rfloor +1) !} \, < \, \epsilon
\end{equation}
On each j-line, there are $(2j+1)$ terms to sum and therefore, the total number of terms it suffices to consider to attain the desired precision is:
\begin{equation}
\sum\limits_{j=0}^{j_\epsilon} \, (2j+1) \, = \, (j_\epsilon + 1)^2
\end{equation}
In table \ref{fig:jepsilon} we define $M_{j_\epsilon}$ to be the l.h.s. of \eqref{Speed3} and we  fix a typical set of market parameters ($S=4200$, $K=4000$, $\sigma = 20 \%$, $r=1\%$ and $\tau = 1Y$). We compute the values of $M_{j_\epsilon}$ for $j_\epsilon\geq 2$ and deduce which precision is attained and after how many terms.

\begin{table}[h!]
\centering
\begin{scriptsize}
\begin{tabular}{|c|ccc|}
\hline
$j_\epsilon$ & $M_{j_\epsilon}$ & Attained precision ($\epsilon$) & Total number of terms $(j_\epsilon + 1)^2$  \\
\hline 
2 & $0.0002258$ & $10^{-2}$  & 9 \\
3 & $0.0000170$ & $10^{-3}$ &  16 \\
4 & $4.27 \times 10^{-7}$ & $10^{-6}$ & 25  \\
5 & $3.21 \times 10^{-8}$ & $10^{-7}$ &  36 \\
6 & $6.03 \times 10^{-10}$ & $10^{-9}$ & 49  \\
7 & $4.54 \times 10^{-11}$ & $10^{-10}$ & 64  \\
  \hline
\end{tabular}
\end{scriptsize}
\caption{Decrease of the bound $M_{j_\epsilon}$ (choice of parameters: $S=3800, \, K=4000, \, r=1\%, \sigma=20\%, \, \tau=1Y$). Only $16$ terms are needed to attain a precision of $10^{-3}$, and 64 for a $10^{-10}$ precision.}
\label{fig:jepsilon}
\end{table}

In table \ref{fig:series1} we plot the first terms of the series representation \eqref{BS_series_3} under the form of lower triangular $(j,n)$-matrix (the first entry of the matrix is the $\frac{1}{2}(S-F)$ term) for a similar set of parameters.

\begin{table}[h!]
\centering
\begin{scriptsize}
\begin{tabular}{|c||ccccccccc|}
  \hline
 {(j,n)} terms  &  & 0 & 1 & 2 & 3 & 4 & 5 & 6 & 7   \\
  \hline
  \hline
    & 119.900 &         &        &         &      &      &  & &     \\
  0 &         & 315.978 & 0      &  0     & 0     & 0     & 0 & 0 & 0 \\
  1 &         & 4.213   & 12.257 & 5.943  & 0     &  0     & 0  & 0 & 0 \\
  2 &         & 0.034   & 0.163  & 0.238  & 0.077 & -0.019 & 0 & 0 & 0  \\
  3 &         & 0.000   & 0.001  & 0.003  & 0.003 &  0.001 &  -0.000  & 0.000 & 0\\
  \hline
 Price & 119.900 & 440.2125 & 452.546 & 458.73 & 458.81 & 458.792 & 458.792  &  458.792 & 458.792  \\
  \hline
\end{tabular}
\end{scriptsize}
\caption{Table containing the numerical values for the $(j,n)$-term in the series (\ref{BS_series_3}) for the option price ($S=4200, \, K=4000, \, r=1\%, \sigma=20\%, \, \tau=1Y$). The call price converges to a precision of $10^{-3}$ after summing only very few terms of the series.}
\label{fig:series1}
\end{table}

\section{Approximations and Hedge parameters}

\subsection{At-the-money price}
The asset is said to be at the money forward if $S=F$ (and therefore $k=0$). In this case, the series \eqref{BS_series_3} becomes
\begin{equation} \label{Brenner_series}
C_{ATMF} \, = \, \frac{S}{2} \sum\limits_{\substack{j\geq 0 \\ n\leq 2j}}  \, \frac{(-1)^n}{n!\Gamma(\frac{3}{2}+j-n)}  \, Z^{2j+1}
\end{equation}
This series is now a series of positive powers of $Z$. It starts for $n=0, m=1$ and goes as follows:
\begin{equation}
C_{ATMF} \, = \, \frac{S}{2}\frac{1}{\Gamma(\frac{3}{2})} \, Z \, + O(Z^3) 
\end{equation}
Recalling that $Z=\frac{\sigma\sqrt{\tau}}{\sqrt{2}}$ and that $\Gamma(\frac{3}{2}) = \frac{\sqrt{\pi}}{2}$ \cite{Abramowitz72}, we get
\begin{equation}
C_{ATMF} \, = \, \frac{S}{\sqrt{2\pi}} \, \sigma\sqrt{\tau} \, + O((\sigma \sqrt{\tau})^3)  
\end{equation}
As:
\begin{equation}
\frac{1}{\sqrt{2\pi}} \, \simeq \, 0.399 \, \simeq \, 0.4
\end{equation}
we thus recover the well-known Brenner-Subrahmanyam approximation \cite{BS94}:
\begin{equation}
C_{ATMF} \, \simeq \, 0.4 \, S \,  \sigma\sqrt{\tau} 
\end{equation}
% Note that the series \eqref{Brenner_series} is indeed a refinement of the Brenner-Subrahmanyam approximation. In this precise market situation, the call price can therefore be expressed as a power series of $\sigma\sqrt{\tau}$, which is not the case in the general case as shown in \eqref{BS_series}, where negative powers of $\sigma\sqrt{\tau}$ and powers of $[\log]$ arise.

\subsection{Greeks}

The hedge parameters, known in quantitative finance as the "Greeks", are derivatives of the option price with respect to time or to market parameters (spot price, risk-free rate, market volatility); they are used to measure the sensibilities of the option to those parameters and to construct appropriate hedging policies. 

The boundary \eqref{Speed} shows that the series \eqref{BS_series_3} converges normally to the call price $C$, and therefore the Greeks can be easily obtained by a term-by-term differentiation of this series.

\subsubsection{Delta}
By definition of $k$ we have $\frac{\partial k}{\partial S} = \frac{1}{S}$ and therefore, by differentiating \eqref{BS_series_3} with respect to $S$ and rearranging the terms:
\begin{equation}
\frac{\partial C}{\partial S} \, = \, \frac{1}{2} \, - \, \frac{1}{2}\frac{F}{S} \, \sum_{\substack{j \geq 1 \\ n \leq 2j-1}} \, Z^{2j-1} \, \frac{(-1)^{n+1}}{n!\Gamma(\frac{1}{2}+j-n)} \, \left( 1-\frac{k}{Z^2} \right)^{n}
\end{equation}
In the ATM-forward configuration ($S=F$, $k=0$) we are left with:
\begin{align}
\frac{\partial C}{\partial S} \, & = \, \frac{1}{2} \, - \, \frac{1}{2} \, \sum_{\substack{j \geq 1 \\ n \leq 2j-1}} \, Z^{2j-1} \, \frac{(-1)^{n+1}}{n!\Gamma(\frac{1}{2}+j-n)} \\
 & = \, \frac{1}{2} \, + \,  \frac{1}{2\sqrt{\pi}}Z \, + \, O(Z^3)
\end{align}

\subsubsection{Rho}
By definition of $F$ and $k$ we have $\frac{\partial F}{\partial r} = -\tau F$ and $\frac{\partial k}{\partial r} = \tau$ and therefore, by differentiating \eqref{BS_series_3} with respect to $r$:
\begin{equation}
\frac{\partial C}{\partial r} \, = \, \frac{1}{2}\tau F \left[ 1 - \,  \sum_{\substack{j \geq 0 \\ n \leq 2j}} \, Z^{2j-1} (Z^2-k+n) \, \frac{(-1)^{n}}{n!\Gamma(\frac{3}{2}+j-n)} \, \left( 1-\frac{k}{Z^2} \right)^{n-1} \right] 
\end{equation}
In the ATM-forward configuration ($S=F$, $k=0$) we are left with:
\begin{align}
\frac{\partial C}{\partial r} \, & = \, \frac{1}{2}\tau F \, \left[ 1  \, - \,  \sum_{\substack{j \geq 0 \\ n \leq 2j}} \, Z^{2j-1} (Z^2+n) \, \frac{(-1)^{n}}{n!\Gamma(\frac{3}{2}+j-n)} \right]   \\
 & = \, \frac{1}{2} \tau F \left[ 1 \,  - \,  \frac{1}{\sqrt{\pi}}Z \, + \, O(Z^2)  \right]
\end{align}

\subsubsection{Vega}
By definition of $Z$ we have $\frac{\partial Z}{\partial \sigma} = \sqrt{\frac{\tau}{2}}$ and therefore, by differentiating \eqref{BS_series_3} with respect to $\sigma$ :
\begin{equation}
\frac{\partial C}{\partial \sigma} \, = \, \frac{F}{2} \, \sqrt{\frac{\tau}{2}} \,  \sum_{\substack{j \geq 0 \\ n \leq 2j}} \, Z^{2(j-1)} \, \frac{(-1)^{n}}{n!\Gamma(\frac{3}{2}+j-n)} \, P_{j,n} \,  \left( 1-\frac{k}{Z^2} \right)^{n-1}
\end{equation}
where
\begin{equation}
P_{j,n} \, = \, (Z^2-k)(1+2j) \, + \, 2 nk
\end{equation}
In the ATM-forward configuration ($S=F$, $k=0$) then $P_{j,n} = (1+2j)Z^2$ and we are left with:
\begin{align}
\frac{\partial C}{\partial \sigma} \, & = \, \frac{S}{2}\sqrt{\frac{\tau}{2}} \, \sum_{\substack{j \geq 0 \\ n \leq 2j}} \, (1+2j) \, Z^{2j} \, \frac{(-1)^{n}}{n!\Gamma(\frac{3}{2}+j-n)} \\
 & = \, S\sqrt{\frac{\tau}{2}} \left[ \frac{1}{\sqrt{\pi}} \, - \, \frac{1}{4\sqrt{\pi}}Z^2 \, + \, O(Z^4) \right]
\end{align}

\subsubsection{Theta}
By definition of $Z$, $F$ and $k$ we have $\frac{\partial Z}{\partial \tau} = \frac{\sigma}{2\sqrt{2\tau}}$, $\frac{\partial F}{\partial \tau} = -rF$, $\frac{\partial k}{\partial \tau} = r$ and therefore, by differentiating \eqref{BS_series_3} with respect to $\tau$ :
\begin{equation}
\frac{\partial C}{\partial \tau} \, = \, \frac{1}{2} r F \, + \, \sigma^2 F \,   \sum_{\substack{j \geq 0 \\ n \leq 2j}} \, Z^{2j-3} \, \frac{(-1)^{n}}{n!\Gamma(\frac{3}{2}+j-n)} \, Q_{j,n} \,  \left( 1-\frac{k}{Z^2} \right)^{n-1}
\end{equation}
where
\begin{equation}
Q_{j,n} \, = \, \frac{1}{8} \, \left[ (1+2j-2r\tau)( Z^2 - k ) + 2n (k-r\tau)   \right]
\end{equation}
In the ATM-forward configuration ($S=F$, $k=0$) then $Q_{j,n} = \frac{1}{8} \, [ (1+2j-2r\tau) Z^2  - 2 n r \tau  ] $ and we are left with:
\begin{align}
\frac{\partial C}{\partial \tau} \, & = \, \frac{1}{2} r F \, + \, \sigma^2 F \,   \sum_{\substack{j \geq 0 \\ n \leq 2j}} \, Z^{2j-3} \, \frac{(-1)^{n} ((1+2j-2r\tau)Z^2 - 2 n r \tau)}{8n!\Gamma(\frac{3}{2}+j-n)} \,  \\
 & = \, \frac{1}{2} r F \, + \, \sigma^2 F \, \left[ \frac{1-r\tau}{4 \sqrt{\pi}} \frac{1}{Z} \,  + \, O(Z) \right]
\end{align}

\section{Concluding remarks}

In this paper, we have provided two new formulas for European options prices in the Black-Scholes model, under the form of rapidly converging series whose terms are straightforward to calculate: formulas \eqref{BS_series} and \eqref{BS_series_3}. Both are simple, but \eqref{BS_series_3} is more tractable for practical applications thanks to his factorized structure in terms of odd powers of $Z$; as this series is uniformly convergent, we were also able to differentiate it term by term, and obtain series expansions for the hedge parameters. Besides their novelty and own mathematical interest, these series can be very easily used in practice and allow to control the calculations to an arbitrary order of precision.  

Let us mention that the multidimensional residue technique we have used to derive these results actually applies to a much wider class of option pricing models, namely all models driven by a space-time fractional generalization of the diffusion equation \eqref{heat_equation}. This class includes not only Black-Scholes but also Finite Moment Log Stable (FMLS) or time-fractional models (see more details in \cite{ACK18}); this is however a particular feature of the Black-Scholes model to be able to transform the residues series into a $Z$-ordered series, and it would not be possible if the options prices were driven by a generic fractional diffusion (only a residues series similar to \eqref{BS_series} could be written down). 

Beyond quantitative finance, the multidimensional residue technique is suitable for the analysis of a PDE as soon as its Green function and its initial or terminal conditions can be expressed under the form of a Mellin-Barnes integral. As illustrated by this paper, this method allows to obtain very precise asymptotic results, without having to solve the PDE explicitly or to implement numerical schemes.

\newpage
\appendix
\section{APPENDIX: Mellin transforms and residues}

We briefly present here some of the concepts used in the paper. The theory of the one-dimensional Mellin transform is explained in full detail in \cite{Flajolet95}. An introduction to multidimensional complex analysis can be found in the classic textbook \cite{Griffiths78}, and applications to the specific case of Mellin-Barnes integrals is developped in \cite{Tsikh94,Tsikh96,Tsikh97}.

\subsection{One-dimensional Mellin transforms}

1. The Mellin transform of a locally continuous function $f$ defined on $\mathbb{R}^+$ is the function $f^*$ defined by
\begin{equation}\label{Mellin_def}
f^*(s) \, := \, \int\limits_0^\infty \, f(x) \, x^{s-1} \, \id x
\end{equation}
The region of convergence $\{ \alpha < Re (s) < \beta \}$ into which the integral \eqref{Mellin_def} converges is often called the fundamental strip of the transform, and sometimes denoted $ < \alpha , \beta  > $.
\newline
\noindent 2. The Mellin transform of the exponential function is, by definition, the Euler Gamma function:
\begin{equation}
\Gamma(s) \, = \, \int\limits_0^\infty \, e^{-x} \, x^{s-1} \, \id x
\end{equation}
with strip of convergence $\{ Re(s) > 0 \}$. Outside of this strip, it can be analytically continued, expect at every negative $s=-n$ integer where it admits the singular behavior
\begin{equation}\label{sing_Gamma}
\Gamma(s) \, \underset{s\rightarrow -n}{\sim} \, \frac{(-1)^n}{n!}\frac{1}{s+n} \hspace*{1cm} n\in\mathbb{N}
\end{equation}
\newline
\noindent 3. The inversion of the Mellin transform is performed via an integral along any vertical line in the strip of convergence:
\begin{equation}\label{inversion}
f(x) \, = \, \int\limits_{c-i\infty}^{c+i\infty} \, f^*(s) \, x^{-s} \, \frac{\id s}{2i\pi} \hspace*{1cm} c\in ( \alpha, \beta )
\end{equation}
and notably for the exponential function one gets the so-called \textit{Cahen-Mellin integral}:
\begin{equation}\label{Cahen}
e^{-x} \, = \, \int\limits_{c-i\infty}^{c+i\infty} \, \Gamma(s) \, x^{-s} \, \frac{\id s}{2i\pi} \hspace*{1cm} c>0
\end{equation}
\newline
\noindent 4. When $f^*(s)$ is a ratio of products of Gamma functions of linear arguments:
\begin{equation}
f^*(s) \, = \, \frac{\Gamma(a_1 s + b_1) \dots \Gamma(a_n s + b_n)}{\Gamma(c_1 s + d_1) \dots \Gamma(c_m s + d_m)}
\end{equation}
then one speaks of a \textit{Mellin-Barnes integral}, whose \textit{characteristic quantity} is defined to be
\begin{equation}
\Delta \, = \, \sum\limits_{k=1}^n \, a_k \, - \, \sum\limits_{j=1}^m \, c_j
\end{equation}
$\Delta$ governs the behavior of $f^*(s)$ when $|s|\rightarrow \infty$ and thus the possibility of computing \eqref{inversion} by summing the residues of the analytic continuation of $f^*(s)$ right or left of the convergence strip:
\begin{equation}
\left\{
\begin{aligned}
& \Delta < 0 \hspace*{1cm} f(x) \, = \, -\sum\limits_{Re(s_N) > \beta} \, \res_{S_N} f^*(s)x^{-s}  \\
& \Delta > 0 \hspace*{1cm} f(x) \, = \, \sum\limits_{Re(s_N) < \alpha} \, \res_{S_N} f^*(s)x^{-s}
\end{aligned}
\right.
\end{equation}
For instance, in the case of the Cahen-Mellin integral one has $\Delta = 1$ and therefore:
\begin{equation}
e^{-x} \, = \, \sum\limits_{Re(s_n)<0} \res_{s_n} \Gamma(s) \, x^{-s} \, = \, \sum\limits_{n=0}^{\infty} \, \frac{(-1)^n}{n!}x^n
\end{equation}
as expected from the usual Taylor series of the exponential function.

\subsection{Multidimensional Mellin transforms}

1. Let the $\underline{a}_k$, $\underline{c}_j$, be vectors in $\mathbb{C}^2$,and the $b_k$, $d_j$ be complex numbers. Let $\underline{t}:=\begin{bmatrix} t_1 \\ t_2 \end{bmatrix}$ and $\underline{c}:=\begin{bmatrix} c_1 \\ c_2 \end{bmatrix}$ in $\mathbb{C}^2$ and "." represent the euclidean scalar product. We speak of a Mellin-Barnes integral in $\mathbb{C}^2$ when one deals with an integral of the type
\begin{equation}
\int\limits_{\underline{c}+i\mathbb{R}^2} \, \omega
\end{equation}
where $\omega$ is a complex differential 2-form who reads
\begin{equation}
\omega \, = \, \frac{\Gamma(\underline{a}_1.\underline{t}_1 + b_1) \dots \Gamma(\underline{a}_n.\underline{t}_n + b_n)}{\Gamma(\underline{c}_1.\underline{t}_1 + d_1) \dots \Gamma(\underline{c}_m.\underline{t}_m + b_m)} \, x^{-t_1} \, y^{-t_2} \, \frac{\id t_1}{2i\pi} \wedge \frac{\id t_2}{2i\pi} \hspace*{1cm} \, x,y \in\mathbb{R}
\end{equation}
The singular sets induced by the singularities of the Gamma functions
\begin{equation}
D_k \, := \, \{ \underline{t}\in\mathbb{C}^2 \, , \, \underline{a}_k.\underline{t}_k + b_k = -n_k \, , \, n_k \in\mathbb{N}   \} \,\,\,\, \, k=0 \dots n
\end{equation}
are called the \textit{divisors} of $\omega$. The \textit{characteristic vector} of $\omega$ is defined to be
\begin{equation}\label{Delta_n}
\Delta \, = \, \sum\limits_{k=1}^n \underline{a}_k \, - \, \sum\limits_{j=1}^m \underline{c}_j
\end{equation}
and the \textit{admissible half-plane}:
\begin{equation}
\Pi_\Delta \, := \, \{ \underline{t}\in\mathbb{C}^2 \, , \, \Delta . \underline{t} \, < \, \Delta . \underline{c}  \}
\end{equation}
\newline
\noindent 2. Let the $\rho_k$ in $\mathbb{R}$, the $h_k:\mathbb{C}\rightarrow\mathbb{C}$ be linear aplications and $\Pi_k$ be a subset of $\mathbb{C}^2$ of the type
\begin{equation}\label{Pik}
\Pi_k \, := \, \{ \underline{t}\in\mathbb{C}^2, \, Re(h_k(t_k)) \, < \, \rho_k \}
\end{equation}
A \textit{cone} in $\mathbb{C}^2$ is a cartesian product
\begin{equation}
\Pi \, = \, \Pi_1 \times \Pi_2
\end{equation}
where $\Pi_1$ and $\Pi_2$ are of the type \eqref{Pik}. Its \textit{faces} $\varphi_k$ are
\begin{equation}
\varphi_k \, := \, \partial \Pi_k \hspace*{1cm} k=1,2
\end{equation}
and its \textit{distinguished boundary}, or \textit{vertex} is
\begin{equation}
\partial_0 \, \Pi \, := \, \varphi_1 \, \cap \, \varphi_2
\end{equation}
\newline
3. Let $1<n_0<n$. We group the divisors $D=\cup_{k=0}^n \, D_k$ of the complex differential form $\omega$ into two sub-families
\begin{equation}
D_1 \, := \, \cup_{k=1}^{n_0} \, D_k \,\,\, \,\,\, D_2 \, := \, \cup_{k=n_0+1}^{n} \, D_k  \hspace*{1cm}  D \, = \, D_1\cup D_2
\end{equation}
We say that a cone $\Pi\subset\mathbb{C}^2$ is \textit{compatible} with the divisors family $D$ if:
\begin{enumerate}
\item[-] \, Its distinguished boundary is $\underline{c}$;
\item[-] \, Every divisor $D_1$ and $D_2$ intersect at most one of his faces:
\begin{equation}
D_k \, \cap \, \varphi_k \, = \, \emptyset \hspace*{1cm} k=1,2
\end{equation}
\end{enumerate}

\noindent 4. Residue theorem for multidimensional Mellin-Barnes integral \cite{Tsikh94,Tsikh97}: If $\Delta \neq 0$ and if $\Pi\subset\Pi_\Delta$ is a compatible cone located into the admissible half-plane, then
\begin{equation}\label{Tsikh_residue}
\int\limits_{\underline{c}+i\mathbb{R}^2} \, \omega \, = \, \sum\limits_{\underline{t}\in\Pi\cap (D_1\cap D_2)} \res_{\underline{t}} \, \omega
\end{equation}
and the series converges absolutely. The residues are to be understood as the "natural" generalization of the Cauchy residue, that is:
\begin{equation}
\left. \res_{\underline{0}} \, \left[ f(t_1,t_2) \, \frac{\id t_1}{2i\pi t_1^{n_1}} \wedge \frac{\id t_2}{2i\pi t_1^{n_2}}  \right] \, = \, \frac{1}{(n_1-1)!(n_2-1)!}\frac{\partial ^{n_1+n_2-2} f}{\partial t_1^{n_1-1}   \partial t_2^{n_2-1} }\right\vert_{\substack{t_1=0\\t_2=0}}
\end{equation}

%----------------------------------------------------------------------------------------
%\end{multicols}
\end{document}